\documentclass[10pt, conference]{IEEEtran}
\usepackage[utf8x]{inputenc}

\usepackage{amsmath,amssymb,amsthm,graphicx,enumerate,verbatim,xcolor}
\usepackage{bbm}
\usepackage{subcaption}
\usepackage{tikz}
\usetikzlibrary{arrows,positioning,fit,backgrounds}
\usetikzlibrary{calc}

\newcommand{\eps}{\varepsilon}

\newcommand{\Ebb}{\mathbb{E}}

\newcommand{\Rbb}{\mathbb{R}}
\newcommand{\Acal}{\mathcal{A}}

\newcommand{\Ecal}{\mathcal{E}}

\newcommand{\Pcal}{\mathcal{P}}

\newcommand{\Xcal}{\mathcal{X}}
\newcommand{\Ycal}{\mathcal{Y}}

\newtheorem{thm}{Theorem}

\newtheorem{lemma}{Lemma}
\newtheorem{defn}{Definition}

\tikzstyle{arw}=[->,>=latex]
\tikzstyle{node}=[draw,rectangle,rounded corners, minimum width=1cm,minimum height =.75 cm]

\title{A Connection between Good Rate-distortion Codes and Backward DMCs}
\author{
\IEEEauthorblockN{Curt Schieler, Paul Cuff}
\vspace{2pt}
\IEEEauthorblockA{Dept. of Electrical Engineering,\\
Princeton University,
Princeton, NJ 08544.\\
E-mail: \{schieler, cuff\}@princeton.edu }
}
\begin{document}
\maketitle
\begin{abstract}
Let ${X^n\in\mathcal{X}^n}$ be a sequence drawn from a discrete memoryless source, and let ${Y^n\in\mathcal{Y}^n}$ be the corresponding reconstruction sequence that is output by a good rate-distortion code. This paper establishes a property of the joint distribution of ${(X^n,Y^n)}$. It is shown that for ${D>0}$, the input-output statistics of a $R(D)$-achieving rate-distortion code converge (in normalized relative entropy) to the output-input statistics of a  discrete memoryless channel (dmc). The dmc is ``backward" in that it is a channel from the reconstruction space $\mathcal{Y}^n$ to source space $\mathcal{X}^n$. It is also shown that the property does not necessarily hold when normalized relative entropy is replaced by variational distance.
\end{abstract}
\section{Introduction}
Consider a discrete memoryless source with generic distribution $P_X$ and a per-symbol distortion measure $d(x,y)$. Given a distortion allowance $D$, the minimum achievable rate of compression (in bits per source symbol) is given by rate-distortion theory as
\begin{equation*}
R(D)=\min_{P_{XY}\in \Pcal(D)} I(X;Y),
\end{equation*}
where
\begin{equation*}
\Pcal(D)=\Big\{P_{XY}: \sum_y P_{XY} =P_X \mbox{ and } \Ebb\,d(X,Y)\leq D\Big\}.
\end{equation*}

One intriguing achievability proof of this classic theorem was given by Wolfowitz in \cite{Wolfowitz1966} (see also \cite[Theorem 7.3]{Csiszar2011}) and goes roughly as follows. 
A joint distribution ${P_{XY}\in\Pcal(D)}$ gives rise to a random transformation $P_{X|Y}$ from the reproduction alphabet to the source alphabet. Using Feinstein's maximal code construction, create a \emph{channel} code designed for the ``backward'' dmc $\prod_{i=1}^n P_{X|Y}(x_i|y_i)$; here, ``backward" refers to the reversed flow of information from the reconstruction space to the source space. The resulting channel code can be transformed into a rate-distortion code by using the channel decoder as a source encoder and the channel encoder as a source decoder. In \cite{Wolfowitz1966}, it is shown that the distortion criterion is met as long as the channel code has large enough error probability, thus demonstrating that good rate-distortion codes can be constructed from certain channel codes. 


\begin{figure}
\centering
 \begin{subfigure}[t]{0.47\textwidth}
 \begin{center}
 \begin{tikzpicture}[node distance=2cm,thick]
 \node (src)   [coordinate] {};
 \node (enc)   [node,right=1.8cm of src] {$f_n$};
 \node (dec)    [node,right=2.5cm of enc] {$g_n$};
 \node (out)  [coordinate,right=1cm of dec] {};

 \draw[arw,thick] (src) to node [midway,above] {$X^n$ i.i.d.} (enc);
 \draw[arw] (enc) to node [midway,above] {$\{1,\ldots,M\}$} (dec);
 \draw[arw] (dec) to node [midway,above] {$Y^n$} (out);
\end{tikzpicture}
\end{center}
\caption{\small A rate-distortion code is a pair $(f_n,g_n)$ that maps a source sequence $X^n$ to a reconstruction codeword $Y^n$. The code induces a distribution $P_{X^nY^n}$ on the pair $(X^n,Y^n)$.}
\label{fig:truedistribution}
\end{subfigure}
 \\
 \vspace{0.5cm}
 \begin{subfigure}[t]{0.47\textwidth}
 \begin{center}
 \begin{tikzpicture}[node distance=2cm,thick]
 \node (src)   [coordinate] {};
 \node (enc)   [node,left=2.7cm of src] {$g_n$};
 \node (dec)    [node,left=1.5cm of enc] {$P_{X|Y}$};
 \node (out)  [coordinate,left=1cm of dec] {};

 \draw[arw,thick] (src) to node [midway,above] {$\mbox{Unif}\{1,\ldots,M\}$} (enc);
 \draw[arw] (enc) to node [midway,above] {$\tilde{Y}^n$} (dec);
 \draw[arw] (dec) to node [midway,above] {$\tilde{X}^n$} (out);
\end{tikzpicture}
\end{center}
\caption{\small Select a codeword $\tilde{Y}^n$ uniformly at random from the codebook corresponding to $(f_n,g_n)$, then pass $\tilde{Y}^n$ through a memoryless channel $P_{X|Y}$. The pair $(\tilde{X}^n,\tilde{Y}^n)$ induces a distribution $Q_{X^nY^n}$. }
\label{fig:approxdistribution}
\end{subfigure}
\caption{\small Description of the true joint distribution $P_{X^nY^n}$ (Fig.~\ref{fig:truedistribution}) and the approximating joint distribution $Q_{X^nY^n}$ (Fig.~\ref{fig:approxdistribution}).}
 \end{figure}
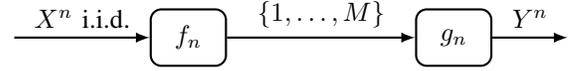
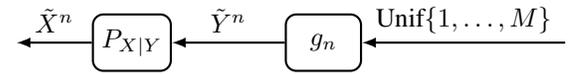

In this paper, we explore another connection between lossy source coding and backward dmc's, one which involves the input-output statistics of good rate-distortion codes. Briefly, the result is as follows. Consider an arbitrary $R(D)$-achieving rate-distortion code\footnote{More precisely, a sequence of codes.} that maps source sequences $X^n$ to reconstruction codewords $Y^n$. The code induces a joint distribution $P_{X^nY^n}$ on the pair $(X^n,Y^n)$ (see Figure~\ref{fig:truedistribution}). 

Using the corresponding codebook, select a codeword uniformly at random as the input to a backward dmc $\prod_{i=1}^n P_{X|Y}(x_i|y_i)$, where $P_{X|Y}$ is derived from the minimizer of $R(D)$.\footnote{Although the minimizer may not be unique, it is well-known that $P_{X|Y}$ is unique.} This channel coding operation induces a joint distribution $Q_{X^nY^n}$ on the pair $(\tilde{X}^n,\tilde{Y}^n)$, where $\tilde{Y}^n$ is the randomly selected codeword and $\tilde{X}^n$ is the channel output (see Figure~\ref{fig:approxdistribution}). We show that, provided some mild necessary conditions are satisfied,
\begin{equation}
\label{introproperty}
\lim_{n\to\infty}\frac1n D(P_{X^nY^n}||Q_{X^nY^n})=0.
\end{equation}
That is, the input-output statistics of nearly all $R(D)$-achieving sequences of rate-distortion codes converge (in the sense of normalized relative entropy) to the output-input statistics of  a backward dmc acting on the rate-distortion codebook.\footnote{We note that a similar claim appears in \cite[Thm. 2]{Pradhan2004}; however, their unconditional claim is not correct. Furthermore, their proof is brief and incorrect.  We comment more on this during our proof. }

The property in \eqref{introproperty} is analogous to the property of capacity-achieving codes for memoryless channels established in~{\cite[Theorem 15]{Han1993}}, namely that the channel output statistics converge (in normalized relative entropy) to a memoryless distribution. More precisely, a capacity-achieving sequence of codes satisfies 
\begin{equation}
\label{intropropertychannel}
\lim_{n\to\infty} \frac1n D(P_{Y^n}||Q_{Y^n})=0,
\end{equation}
where $P_{Y^n}$ is the true distribution of the channel output and ${Q_{Y^n}=\prod_{i=1}^n P_Y(y_i)}$, where $P_Y$ is the unique capacity-achieving output distribution. 

There are various properties of good rate-distortion codes that have been examined in the past (see, for example, \cite{Kanlis1996} and \cite{Weissman2005}). Notably, \cite{Weissman2005} showed that the empirical $k$th-order distribution of a good rate-distortion code converges in distribution almost surely to the minimizer of the $k$th-order rate-distortion function (when that minimizer is unique). Note that the property in \eqref{introproperty}, in contrast, concerns the actual (not empirical) joint distribution and ${k=n}$. In some sense, \eqref{introproperty} complements \cite{Weissman2005} in the same way that \eqref{intropropertychannel} complements the results in \cite{Shamai1997} on the $k$th-order empirical input distribution of good channel codes.

In order to show that good rate-distortion codes yield \eqref{introproperty}, we will first prove in Section~\ref{sec:goodcoordination} that the property holds for good empirical coordination codes. Empirical coordination, studied in \cite{Cuff2010}, is similar to rate-distortion except for the distortion criterion, which is replaced by the requirement that the variational distance between the joint empirical distribution and a target joint distribution $P_{XY}$ converges in probability. Thus, one aims to achieve coordination pairs $(R,P_{Y|X})$ instead of rate-distortion pairs $(R,D)$. Upon demonstrating that \eqref{introproperty} holds for good empirical coordination codes, we show in Section~\ref{sec:gooddistortion} that the property holds for good rate-distortion codes, as well. In Section~\ref{sec:otherdistances}, we show that the property can fail to hold when the distance measure is replaced by variational distance or unnormalized relative entropy. 

Although we do not prove it here, we are able to use the property in \eqref{introproperty} to solve a problem in information-theoretic secrecy relating to Yamamoto's ``Rate-distortion theory of the Shannon cipher system" \cite{Yamamoto1997}. Specifically, one can use the property to show that the results of \cite{Cuff2010partial} can be achieved simply by using good rate-distortion codes, instead of the particular stochastic encoders that \cite{Cuff2010partial} asserts the existence of. It is likely that the property can provide a solution or give insight into other secrecy problems, as well.
\section{Good empirical coordination codes}
\label{sec:goodcoordination}
We begin by introducing empirical coordination codes. All results in this paper will assume memoryless sources and finite alphabets. Furthermore, we assume for simplicity that the source satisfies ${P_X(x)>0,\, \forall x\in\Xcal}$. We first give the definition of a coordination code (see Figure~\ref{fig:truedistribution}).
\begin{defn}
\label{codedefn}
An $(n,R_n)$ coordination code consists of an encoder-decoder pair $(f_n,g_n)$ operating at rate $R_n$, where
\begin{IEEEeqnarray}{l}
f_n: \Xcal^n \longrightarrow \{1,\ldots,M\}\\
g_n:\{1,\ldots,M\}\longrightarrow \Ycal^n\\
R_n=\tfrac1n \log M.
\end{IEEEeqnarray}
\end{defn}
A coordination code acts on a memoryless source $X^n$ with generic distribution $P_X$. For a fixed source sequence $x^n$, the code produces a codeword $y^n=g(f(x^n))$. The empirical distribution of the resulting pair $(x^n,y^n)$ is defined for all $(x,y)\in\Xcal\times\Ycal$ by
\begin{equation}
T_{x^ny^n}(x,y) \triangleq \frac1n \sum_{i=1}^n \bold{1}\{(x_i,y_i)=(x,y)\}.
\end{equation}
The empirical distribution of the pair of random variables $(X^n,Y^n)$ is itself a random variable and is denoted by $T_{X^nY^n}$. Variational distance, a measure of the distance between two distributions $P$ and $Q$ with common alphabet, is defined by
\begin{equation}
\lVert P - Q \rVert \triangleq \sup_{A} |P(A)-Q(A)|.
\end{equation}
\begin{defn}
The pair $(R,P_{Y|X})$ is achievable if there exists a sequence of $(n,R_n)$ coordination codes such that
\begin{equation}
\lim_{n\to\infty} R_n = R
\end{equation}
and
\begin{equation}
\lVert T_{X^nY^n} - P_{XY} \rVert \xrightarrow{i.p.}0,
\end{equation}
where $P_{XY}=P_XP_{Y|X}$.
\end{defn}
\begin{thm}[{\cite{Cuff2010}}]
\label{coordinationthm}
The pair $(R,P_{Y|X})$ is achievable if and only if $R\geq I(X;Y)$.
\end{thm}
The rate boundary in Theorem~\ref{coordinationthm} justifies the following definition of a ``good" coordination code.
\begin{defn}
Given a source $P_X$, a sequence of $(n,R_n)$ coordination codes $\{(f_n,g_n)\}_{n=1}^\infty$ is good for $P_{Y|X}$ if
\begin{equation}
\label{optrate}
\lim_{n\to\infty} R_n = I(X;Y)
\end{equation}
and
\begin{equation}
\label{opttv}
\lVert T_{X^nY^n} - P_{XY} \rVert \xrightarrow{i.p.}0.
\end{equation}
\end{defn}
To each good sequence of coordination codes for $P_{Y|X}$, we associate two sequences of joint distributions $\{P_{X^nY^n}\}_{n=1}^\infty$ and $\{Q_{X^nY^n}\}_{n=1}^\infty$. The first, $P_{X^nY^n}$, is the distribution of the pair $(X^n,Y^n)$ induced by the code. That is,
\begin{equation}
\label{beginP}
P_{X^nY^n}=P_{X^n}P_{Y^n|X^n},
\end{equation}
where
\begin{equation}
P_{X^n}(x^n)=\prod_{i=1}^n P_X(x_i)
\end{equation}
is the memoryless source distribution and
\begin{equation}
P_{Y^n|X^n}(y^n|x^n) = \bold{1}\big\{y^n = g_n(f_n(x^n))\big\}
\end{equation}
is the composition of the encoder with the decoder. The second distribution, $Q_{X^nY^n}$, is the distribution of the pair $(\tilde{X}^n,\tilde{Y}^n)$, where $\tilde{Y}^n$ is a codeword selected uniformly at random and $\tilde{X}^n$ is the output of the backward dmc when the input is $\tilde{Y^n}$. That is,
\begin{equation}
Q_{X^nY^n}=Q_{Y^n}Q_{X^n|Y^n},
\end{equation}
where
\begin{equation}
Q_{Y^n}(y^n)=\frac{| g_n^{-1}(y^n) |}{M} 
\end{equation}
is the uniform distribution over the codebook (which might contain duplicate codewords) and
\begin{equation}
\label{endQ}
Q_{X^n|Y^n}(x^n|y^n) = \prod_{i=1}^n P_{X|Y}(x_i|y_i)
\end{equation}
is the backward dmc with generic channel $P_{X|Y}$ derived from the joint distribution $P_{XY}=P_XP_{Y|X}$. 

Our main result is the following theorem. \footnote{We exclude the single pathological case $P_{XY}(x,y)=\frac{1}{|\Xcal|}\bold{1}\{x=y\}$, in which it is possible that there are some codebooks such that $P_{X^nY^n}=Q_{X^nY^n}$ and other codebooks such that $D(P_{X^nY^n}||Q_{X^nY^n})=\infty$.}
\begin{thm}
\label{mainresultcoordination}
Let ${P_{XY}\in \Acal}$, where
\begin{equation}
\Acal\triangleq \{P_{XY}: P_{X|Y}(x|y) > 0, \forall (x,y)\}
\end{equation}
Then, for any good sequence of coordination codes for $P_{Y|X}$, it holds that
\begin{equation}
\lim_{n\to\infty} \frac1n D(P_{X^nY^n}||Q_{X^nY^n})=0,
\end{equation}
where $P_{X^nY^n}$ and $Q_{X^nY^n}$ are defined in \eqref{beginP}-\eqref{endQ}. Furthermore, if ${P_{XY}\notin \Acal}$, then there exists a good sequence of coordination codes for $P_{Y|X}$ such that
\begin{equation}
\lim_{n\to\infty} \frac1n D(P_{X^nY^n}||Q_{X^nY^n})=\infty.
\end{equation}
\end{thm}
\vspace{3mm}
\begin{proof}

We will need the following property of variational distance, which is easily verified. Let ${\eps>0}$ and let $f(x)$ be a function bounded by ${b\in\Rbb}$. Then
\begin{equation}
\label{tvcontinuous}
\lVert P-Q \rVert < \eps \:\Longrightarrow\: \big| \Ebb_Pf(X) - \Ebb_Qf(X) \big | < \eps b.
\end{equation}

We also need the following chain rule of relative entropy:
\begin{IEEEeqnarray}{rCl}
 \IEEEeqnarraymulticol{3}{l}{\nonumber
D(P_{X^nY^n}||Q_{X^nY^n})
}\\
\label{divergencechainrule}\quad &=& D(P_{X^nY^n}||P_{Y^n}Q_{X^n|Y^n})+D(P_{Y^n}||Q_{Y^n}).
\end{IEEEeqnarray}

To begin the proof of Theorem~\ref{mainresultcoordination}, fix ${P_{XY}\in \Acal}$ and a good sequence of coordination codes for $P_{Y|X}$. We first show that such a sequence has the property\footnote{In \cite{Pradhan2004}, the assertion is that the theorem follows from \eqref{propertymutualinfo}. However, this is not the case. It is necessary to establish the steps in \eqref{crucialproperty}-\eqref{intermediatemutualinfo}, which rely on the property of coordination codes in \eqref{tvvanish}.}
\begin{equation}
\label{propertymutualinfo}
\lim_{n\to\infty} \frac1n I(X^n;Y^n)=I(X;Y),
\end{equation}
where $I(X^n;Y^n)$ is evaluated with respect to the true distribution $P_{X^nY^n}$. Throughout the proof, bear in mind that all expectations and mutual information expressions involving $(X^n,Y^n)$ are evaluated with respect to the true distribution $P_{X^nY^n}$.

To show \eqref{propertymutualinfo}, we first introduce an auxiliary random variable ${J\sim \mbox{Unif}\{1,\ldots,n\}}$ independent of $(X^n,Y^n)$. Regurgitating some of the standard steps found in the converse to the lossy source coding theorem, we have
 \begin{IEEEeqnarray}{rCl}
 R_n &=& \frac1n \log M \\
 &\geq& \frac1n H(Y^n)\\
 &\geq& \frac1n I(X^n;Y^n)\\
 &\geq& \frac1n \sum_{i=1}^n I(X_i,Y_i)\\
 &=& I(X_J;Y_J|J)\\
 &\stackrel{(a)}{=}& I(X_J;Y_J,J)\\
 &\geq& I(X_J;Y_J),
 \end{IEEEeqnarray}
 where $(a)$ follows from $X_J \perp J$. If we can show that
\begin{equation}
\label{mutualinfoconverge}
\lim_{n\rightarrow\infty}I(X_J;Y_J)=I(X;Y),
\end{equation}
then the proof of the property in \eqref{propertymutualinfo} will be complete by \eqref{optrate} and the squeeze theorem. To that end, we use several observations from \cite{Cuff2010}. By the boundedness of variational distance, \eqref{opttv} implies
\begin{equation}
\label{expectedtvconverge}
\lim_{n\to\infty}\Ebb\,\lVert T_{X^nY^n}-P_{XY}\rVert= 0.
\end{equation}
Upon noting that
\begin{equation}
\Ebb\,T_{X^nY^n}=P_{X_JY_J},
\end{equation}
we have
\begin{IEEEeqnarray}{rCl}
 \lVert P_{X_JY_J} - P_{XY} \rVert &=& \lVert \Ebb\,T_{X^nY^n} - P_{XY}\rVert\\
 &\stackrel{(a)}{\leq}& \Ebb\,\lVert T_{X^nY^n}-P_{XY} \rVert,
\end{IEEEeqnarray}
where (a) follows from Jensen's inequality. Therefore,
\begin{equation}
\label{tvvanish}
\lim_{n\to\infty} \lVert P_{X_JY_J} - P_{XY} \rVert = 0.
\end{equation}
Since mutual information is continuous with respect to variational distance for finite alphabets (this follows from \eqref{tvcontinuous}), we see that \eqref{tvvanish} yields \eqref{mutualinfoconverge}. Thus, the property in \eqref{propertymutualinfo} holds.

We remark that the property in \eqref{tvvanish} underlies the reason that we are considering empirical coordination codes. In brief, it arises more naturally in an empirical coordination setting than in a rate-distortion setting. We will invoke \eqref{tvvanish} again shortly.

With \eqref{propertymutualinfo} in hand, we now show that
\begin{equation}
 \label{almostthm}\lim_{n\rightarrow\infty} \frac1n D(P_{X^nY^n}||P_{Y^n}Q_{X^n|Y^n})=0.
\end{equation}
 To start, we have
 \begin{IEEEeqnarray}{rCl}
 \IEEEeqnarraymulticol{3}{l}{
\nonumber\lim_{n\to\infty}\frac1n \Ebb\Big[\log \frac{\prod_{i=1}^n P_{X|Y}(X_i|Y_i)}{\prod_{i=1}^n P_{X}(X_i)}\Big]}\\
 \quad &=&   \label{crucialproperty}\lim_{n\to\infty}\frac1n \sum_{i=1}^n \Ebb\Big[\log \frac{P_{X|Y}(X_i|Y_i)}{P_X(X_i)}\Big]\\
 &=& \lim_{n\to\infty}\Ebb\Big[\log\frac{P_{X|Y}(X_J|Y_J)}{P_X(X_J)}\Big]\\
 &\stackrel{(a)}{=}& \lim_{n\to\infty}\Ebb\Big[\log\frac{P_{X|Y}(X|Y)}{P_X(X)}\Big] \\
 \label{intermediatemutualinfo} &=& I(X;Y).
\end{IEEEeqnarray}
To see how (a) follows, first note that the function
\begin{equation}
f(x,y)=\log \frac{P_{X|Y}(x|y)}{P_{X}(x)}
\end{equation}
is bounded due to the restriction ${P_{XY}\in\Acal}$ (in fact, this is the only step where the restriction is needed). Then, use \eqref{tvvanish} along with \eqref{tvcontinuous}. Continuing, we have
\begin{IEEEeqnarray}{rCl}
\IEEEeqnarraymulticol{3}{l}{
 \nonumber\lim_{n\to\infty}\frac1n D(P_{Y^n}P_{X^n|Y^n}||P_{Y^n}Q_{X^n|Y^n})
 }\\
 \:&=& \lim_{n\to\infty}\frac1n \Ebb\Big[ \log \frac{P_{X^n|Y^n}(X^n|Y^n)}{Q_{X^n|Y^n}(X^n|Y^n)} \Big]\\
 &=& \lim_{n\to\infty}\frac1n \Ebb\Big[\log \frac{P_{X^n|Y^n}(X^n|Y^n)}{P_{X^n}(X^n)}\Big]\\
 &&\quad-\,\lim_{n\to\infty}\frac1n \Ebb\Big[\log \frac{Q_{X^n|Y^n}(X^n|Y^n)}{P_{X^n}(X^n)}\Big]\\
 &=& \lim_{n\to\infty}\frac1n I(X^n;Y^n)\\
 &&\quad -\,\lim_{n\to\infty}\frac1n \Ebb\Big[\log \frac{\prod P_{X|Y}(X_i|Y_i)}{\prod P_{X}(X_i)}\Big]\\
 &\stackrel{(a)}{=}& \lim_{n\to\infty}\frac1n I(X^n;Y^n) - I(X;Y)\\
 &\stackrel{(b)}{=}& 0,
\end{IEEEeqnarray}
where (a) is due to \eqref{intermediatemutualinfo} and (b) is due to \eqref{propertymutualinfo}. This proves the property in \eqref{almostthm}. 

Finally, write 
\begin{IEEEeqnarray}{rCl}
\IEEEeqnarraymulticol{3}{l}{
\nonumber\lim_{n\to\infty} \frac1n D(P_{Y^n}||Q_{Y^n})
}\\
\:&=& \lim_{n\to\infty} \frac1n\sum_{y^n} P_{Y^n}(y^n) \log\frac{1}{Q_{Y^n}(y^n)}\\
&&\quad -\,\lim_{n\to\infty} \frac1n H(Y^n)\\
&\leq& \lim_{n\to\infty} \frac1n \log M - \lim_{n\to\infty} \frac1n H(Y^n)\\
&\stackrel{(a)}{=}& 0
\end{IEEEeqnarray}
where (a) follows from the squeeze theorem. To complete the first part of the theorem, invoke the chain rule of relative entropy in \eqref{divergencechainrule}.

To show the second part of Theorem~\ref{mainresultcoordination}, fix $P_{XY}\notin \Acal$ and a good sequence of coordination codes for the corresponding $P_{Y|X}$. The condition $P_{XY}\notin \Acal$ implies the existence of a pair $(x,y)$ such that $P_{X|Y}(x|y)=0$. For every $n$, append a codeword $y^n$ to the codebook and associate with it a sequence $x^n$ such that $$|i:(x_i,y_i)=(x,y)|>0.$$ Accordingly, modify $f_n$ and $g_n$ so that $y^n=g(f(x^n))$. Such a modification maintains the goodness of the code, but now $P_{X^nY^n}$ has support on $(x^n,y^n)$, while $Q_{X^nY^n}$ does not. Consequently, $\frac1n D(P_{X^nY^n}||Q_{X^nY^n})$ diverges. 
\end{proof}
\section{Good rate-distortion codes}
\label{sec:gooddistortion}
In this section, we establish the counterpart to Theorem~\ref{mainresultcoordination} for good rate-distortion codes. A rate-distortion code is defined according to Definition~\ref{codedefn}. The notion of good is also similar; in this case, a good code is an $R(D)$-achieving one.
\begin{defn}
Given a source $P_X$ and a distortion measure $d(x,y)$, a sequence of $(n,R_n)$ rate-distortion codes $\{(f_n,g_n)\}_{n=1}^\infty$ is good for distortion $D$ if 
\begin{equation}
\label{optrate2}
\lim_{n\to\infty} R_n = R(D)
\end{equation}
and
\begin{equation}
\label{optdistortion}
\lim_{n\to\infty} \frac1n \sum_{i=1}^n \Ebb\,d(X_i,Y_i)  \leq D.
\end{equation}
\end{defn}
For a fixed per-letter distortion measure $d(x,y)$, the rate-distortion function is defined for ${D\geq D_{\min}}$, where ${D_{\min}=\Ebb[\min_y d(X,y)]}$. Without loss of generality, we assume that $D_{\min}=0$.

 In view of the restriction in Theorem~\ref{mainresultcoordination} to ${P_{XY}\in\Acal}$, the following lemma is useful.
 \begin{lemma}[{\cite[Ch. 2, Lemma 1]{Berger1971}}]
 \label{reduce}
 Let ${D>0}$. Any $P_{XY}$ that minimizes $R(D)$ is such that, if ${P_{Y|X}(y|x)=0}$ for some $(x,y)$, then ${P_{Y|X}(y|x')=0}$ for all $x'\in \Xcal$. Accordingly, the reproduction symbol $y$ may be deleted from $\Ycal$ without affecting $R(D)$.
 \end{lemma}
Thus, we have that for any ${D>0}$ we can reduce the reproduction alphabet $\Ycal$, without penalty, to an alphabet $\Ycal^*(D)$ such that any $P_{XY}$ minimizing $R(D)$ satisfies ${P_{XY}(x,y)>0}$ for all $(x,y)\in\Xcal\times\Ycal^*$. In particular, $P_{XY}\in\Acal$. It is shown in \cite{Berger1971} that this does not hold for $D=0$. From this point on, we assume that $\Ycal$ has been reduced according to Lemma~\ref{reduce}, so that Theorem~\ref{mainresultcoordination} can be invoked.

Although the minimizer of $R(D)$ need not be unique, it turns out that the corresponding backward channel $P_{X|Y}$ is unique. This is analogous to the fact that the capacity-achieving output distribution is unique, even though the input distribution is not.
\begin{lemma}[{\cite[Problem 8.3]{Csiszar2011}}]
If $P_{XY}$ and $Q_{XY}$ both minimize $R(D)$, then ${P_{X|Y}=Q_{X|Y}}$.
\end{lemma}
We now state the counterpart to Theorem~\ref{mainresultcoordination}. The proof is immediate once we use the fact that good rate-distortion codes are good empirical coordination codes.
\begin{thm}
\label{mainresultdistortion}
Let ${D>0}$, and assume that the reproduction alphabet has been reduced to $\Ycal^*(D)$. 
Then, for any good sequence of rate-distortion codes for $D$, it holds that
\begin{equation}
\lim_{n\to\infty} \frac1n D(P_{X^nY^n}||Q_{X^nY^n})=0,
\end{equation}
where
\begin{equation}
P_{X^nY^n}(x^n,y^n)=\prod_{i=1}^n P_X(x_i)\:\bold{1}\big\{y^n = g_n(f_n(x^n))\big\}
\end{equation}
and
\begin{equation}
Q_{X^nY^n}(x^n,y^n)=\frac{| g_n^{-1}(y^n) |}{M} \prod_{i=1}^n P_{X|Y}(x_i|y_i),
\end{equation}
where $P_{X|Y}$ is the unique backward channel corresponding to $D$.
\end{thm}
\begin{proof}
From \cite[Theorem 11]{Cuff2010} or \cite[Theorem 9]{Weissman2005}, we have that a good rate-distortion code for $D$ is a good empirical coordination code for some $P_{Y|X}$ minimizing $R(D)$. Due to the reduction to $\Ycal^*(D)$, we have $P_{XY}\in\Acal$, which allows us to invoke Theorem~\ref{mainresultcoordination}.
\end{proof}
\section{Variational Distance}
\label{sec:otherdistances}
In this section, we show that Theorem~\ref{mainresultcoordination} does not hold when we replace normalized divergence by variational distance. From Pinsker's inequality, this implies that it does not hold in unnormalized relative entropy, either.
\begin{thm}
There exists $P_{XY}\in\Acal$ and a sequence of good coordination codes for the corresponding $P_{Y|X}$ such that
\begin{equation}
\lim_{n\to\infty}  \lVert P_{X^nY^n} - Q_{X^nY^n} \rVert \neq 0,
\end{equation}
where $P_{X^nY^n}$ and $Q_{X^nY^n}$ are defined in \eqref{beginP}-\eqref{endQ}.
\end{thm}
\begin{proof}
Let ${P_{XY}\in\Acal}$ be such that $P_{Y}$ is an capacity-achieving input distribution of the channel $P_{X|Y}$. Fix a sequence of good empirical coordination codes $\{(f_n,g_n)\}_{n=1}^\infty$ for $P_{Y|X}$ such that the decoder is bijective and
\begin{equation}
\label{slowconverge}
R_n= I(X;Y)+n^{-\frac12 + \delta},
\end{equation}
for some $\delta>0$.
This is possible by Theorem~\ref{coordinationthm}. By way of contradiction, suppose that
\begin{equation}
\label{contradiction}
\lim_{n\to\infty}  \lVert P_{X^nY^n} - Q_{X^nY^n} \rVert = 0.
\end{equation}
To reach a contradiction, we first define joint distributions $P_{X^nY^n\widehat{M}}$ and $Q_{X^nY^n\widehat{M}}$ by
\vspace{5pt}
\begin{IEEEeqnarray}{l}
\nonumber P_{X^nY^n\widehat{M}}(x^n,y^n,\widehat{m})=P_{X^nY^n}(x^n,y^n)\,\bold{1}\big\{\widehat{m}=f_n(x^n)\big\}\vspace{5pt}\\
\nonumber Q_{X^nY^n\widehat{M}}(x^n,y^n,\widehat{m})=Q_{X^nY^n}(x^n,y^n)\,\bold{1}\big\{\widehat{m}=f_n(x^n)\big\}.
\end{IEEEeqnarray}

Observe that $Q_{X^nY^n\widehat{M}}$ is the joint distribution governing the triple $(X^n,Y^n,\widehat{M})$ in the following channel coding setting:
 
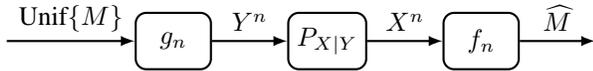
\begin{center}
 \begin{tikzpicture}[node distance=2cm,thick]
 \node (src)   [coordinate] {};
 \node (enc)   [node,right=1.7cm of src] {$g_n$};
 \node (ch)    [node,right=1cm of enc] {$P_{X|Y}$};
 \node (dec) [node,right=1cm of ch] {$f_n$};
 \node (out)  [coordinate,right=1cm of dec] {};

 \draw[arw,thick] (src) to node [midway,above] {$\mbox{Unif}\{M\}$} (enc);
 \draw[arw] (enc) to node [midway,above] {$Y^n$} (ch);
 \draw[arw] (ch) to node [midway,above] {$X^n$} (dec);
 \draw[arw] (dec) to node [midway,above] {$\widehat{M}$} (out);
\end{tikzpicture}
\end{center}
Thus, we have turned the rate-distortion code $(f_n,g_n)$ into a channel code by identifying the channel encoder as the source decoder and the channel decoder as the source encoder. Because $g_n$ is bijective, the error event for the channel coding is given by
\begin{equation}
\Ecal_n=\left\{g_n(\widehat{M})\neq Y^n\right\},
\end{equation}
and the probability of error is $\text{Pr}\{\text{Error}(n)\}\triangleq Q(\Ecal_n)$. 

On the other hand, notice that under the distribution $P_{X^nY^n\widehat{M}}$, it holds that
\begin{equation}
g_n(\widehat{M})=g_n(f_n(X^n))=Y^n,
\end{equation}
and thus $P(\Ecal_n)=0$.

Now, since variational distance has the property
\begin{equation}
\lVert P_X P_{Y|X} - Q_X P_{Y|X} \rVert = \lVert P_X - Q_X \rVert,
\end{equation}
we have by \eqref{contradiction} that
\begin{IEEEeqnarray}{rCl}
\IEEEeqnarraymulticol{3}{l}{
\lim_{n\to\infty}  \lVert P_{X^nY^n\widehat{M}} - Q_{X^nY^n\widehat{M}} \rVert
}\\
\quad &=& \lim_{n\to\infty}  \lVert P_{X^nY^n} - Q_{X^nY^n} \rVert\\
&=& 0.
\end{IEEEeqnarray}
Therefore, by the definition of variational distance,
\begin{IEEEeqnarray}{rCl}
\lim_{n\to\infty} \text{Pr}\{\text{Error}(n)\} &=& \lim_{n\to\infty} Q(\Ecal_n)\\
&=& \lim_{n\to\infty} | Q(\Ecal_n) - P(\Ecal_n) |\\
&=& 0.
\end{IEEEeqnarray}
Thus, we have demonstrated a sequence of channel codes whose rates approach the channel capacity slowly\footnote{Referring to the term $n^{-\frac12+\delta}$ in \eqref{slowconverge}.} from above, yet whose probability of error vanishes. This is impossible due to the strong converse to the channel coding theorem (e.g., \cite[Theorem 5.8.5]{Gallager1968}), yielding a contradiction.
\end{proof}
\section{Acknowledgements}
This research was supported in part by the National Science Foundation under Grants CCF-1116013 and CCF-1017431, and also by the Air Force Office of Scientific Research under Grant FA9550-12-1-0196.
\bibliographystyle{IEEEtran}
\bibliography{backDMC}
\end{document}